\documentclass[letterpaper,aps,prl,onecolumn,notitlepage,nofootinbib,10pt,superscriptaddress,longbibliography]{revtex4-1}
\usepackage[linesnumbered,ruled,vlined]{algorithm2e}
\usepackage{amsmath}
\usepackage{amssymb}
\usepackage{amsfonts}
\usepackage{amsthm}
\usepackage{youngtab}

\usepackage{graphicx}

\usepackage{color}
\definecolor{darkblue}{rgb}{0.,0.,0.4}
\definecolor{darkred}{rgb}{0.5,0.,0.}
\usepackage[pdftex,colorlinks=true,linkcolor=darkblue,citecolor=darkred,urlcolor=darkblue]{hyperref}

\newtheorem{thm}{Theorem}
\newtheorem*{thm*}{Theorem}

\newtheorem{lem}[thm]{Lemma}
\newtheorem{prop}[thm]{Proposition}

\theoremstyle{definition}

\newtheorem{observation}{Observation}

\newtheorem{dfn}[thm]{Definition}
\theoremstyle{plain}
\newtheorem{prob}{Problem}
\makeatletter
\newtheorem*{rep@theorem}{\rep@title}
\newcommand{\newreptheorem}[2]{%
\newenvironment{rep#1}[1]{%
 \def\rep@title{#2 \ref{##1} (restatement)}%
 \begin{rep@theorem}}%
 {\end{rep@theorem}}}
\makeatother

\newreptheorem{thm}{Theorem}

\newcommand{\nc}{\newcommand}
\nc{\rnc}{\renewcommand}

\nc\eps{\epsilon}

\nc\bbC{\mathbb{C}}

\nc\bbF{\mathbb{F}}
\nc\bbM{\mathbb{M}}
\nc\bbN{\mathbb{N}}
\nc\bbR{\mathbb{R}}
\nc\bbS{\mathbb{S}}
\nc\bbZ{\mathbb{Z}}

\nc\bp{\mathbf{p}}
\nc\bq{\mathbf{q}}

\nc\benum{\begin{enumerate}}
\nc\eenum{\end{enumerate}}
\nc\bit{\begin{itemize}}
\nc\eit{\end{itemize}}

\nc{\todo}[1]{\textcolor{red}{todo: #1}}

\nc{\Anote}[1]{\textcolor{red}{Aram note: #1}}
\nc{\Ynote}[1]{\textcolor{red}{Nengkun note: #1}}

\nc\cA{\mathcal{A}}
\nc\cB{\mathcal{B}}
\nc\cC{\mathcal{C}}
\nc\cD{\mathcal{D}}
\nc\cE{\mathcal{E}}
\nc\cF{\mathcal{F}}
\nc\cG{\mathcal{G}}
\nc\cH{\mathcal{H}}
\nc\cI{\mathcal{I}}
\nc\cJ{\mathcal{J}}
\nc\cK{\mathcal{K}}
\nc\cL{\mathcal{L}}
\nc\cM{\mathcal{M}}
\nc\cN{\mathcal{N}}
\nc\cO{\mathcal{O}}
\nc\cP{\mathcal{P}}
\nc\cQ{\mathcal{Q}}
\nc\cR{\mathcal{R}}
\nc\cS{\mathcal{S}}
\nc\cT{\mathcal{T}}
\nc\cU{\mathcal{U}}
\nc\cV{\mathcal{V}}
\nc\cW{\mathcal{W}}
\nc\cX{\mathcal{X}}
\nc\cY{\mathcal{Y}}
\nc\cZ{\mathcal{Z}}

\DeclareMathOperator{\tr}{tr}

\def\be#1\ee{\begin{equation}#1\end{equation}}
\def\bea#1\eea{\begin{eqnarray}#1\end{eqnarray}}
\def\beas#1\eeas{\begin{eqnarray*}#1\end{eqnarray*}}
\def\ba#1\ea{\begin{align}#1\end{align}}
\def\bas#1\eas{\begin{align*}#1\end{align*}}
\def\bpm#1\epm{\begin{pmatrix}#1\end{pmatrix}}

\rnc\L{\left}
\nc\R{\right}
\nc\ra{\rightarrow}
\nc\ot{\otimes}

\begin{document}

\title{Sample efficient tomography via Pauli Measurements}

\author{Nengkun Yu}

\affiliation{Centre for Quantum Software and Information, Faculty of Engineering and Information Technology, University of Technology, Sydney, NSW 2007, Australia}

\date{\today}


\begin{abstract}
Pauli Measurements are the most important measurements in both theoretical and experimental aspects of quantum information science. In this paper, we explore the power of Pauli measurements in the state tomography related problems. Firstly, we show that the \textit{quantum state tomography} problem of $n$-qubit system can be accomplished with ${\mathcal{O}}(\frac{10^n}{\epsilon^2})$ copies of the unknown state using Pauli measurements. As a direct application, we studied the \textit{quantum overlapping tomography} problem introduced by Cotler and Wilczek in Ref.~\cite{Cotler_2020}. We show that the sample complexity is $\mathcal{O}(\frac{10^k\cdot\log({{n}\choose{k}}/\delta))}{\epsilon^{2}})$ for quantum overlapping tomography of $k$-qubit reduced density matrices among $n$ is quantum system, where $1-\delta$ is the confidential level, and $\epsilon$ is the trace distance error. This can be achieved using Pauli measurements. Moreover, we prove that $\Omega(\frac{\log(n/\delta)}{\epsilon^{2}})$ copies are needed. In other words, for constant $k$, joint, highly entangled, measurements are not asymptotically more efficient than Pauli measurements. 

\end{abstract}

\pacs{03.65.Wj, 03.67.-a}

\maketitle
\section{Introduction}
It is increasingly important to understand how the cost scales of learning useful information as the experiments can control larger and larger quantum systems.
Quantum state tomography refers to a procedure of reconstructing the
  density matrix of a quantum state from various measurements on
  multiple copies of the state. This fundamental task is crucial for
  quantum information experiments and theoretically goes back at least
  to the work of Helstrom, Holevo, and others from around 1970.
  
In \textit{quantum state tomography}, one is given $k$ copies of an $n$-qubit quantum
state $\rho$ and is required to output a classical description of
density matrix $\hat{\rho}$ close to $\rho$ by performing quantum
measurements on $\rho^{\otimes k}$. Broadly speaking, there are three
categories of measurements: one consists of joint (entangled)
measurements on $\rho^{\otimes k}$, as in \cite{BBMR04,Keyl06,GJK08}.
In \cite{HHJ+16,OW16,OW17}, the authors showed that the optimal scaling of the sample number $k$ as a
function of trace distance goal $\eps$ is $n \propto \frac{4^n}{\eps^2}$. The second category
consists of measurements on each copy of the state $\rho$,
whose results are to be combined to reconstruct the
density matrix, as in \cite{FlammiaGrossLiuEtAl2012}. Ref.~\cite{KRT14} 
showed that if one can perform many-outcome measurements, 
tomography is possible using $n \propto 8^n / \delta^2$ copies, and this is optimal if the measurements on every copy are independent as showed in \cite{HHJ+16}.
The third category local measurements consists of measurements on each qubit of each copy of the state $\rho$, whose results are to be combined to reconstruct the
density matrix \cite{Gu__2020}. Generally, joint measurements over several copies of the state can usually achieve lower sampling rates \cite{HHJ+16,OW16,OW17}, but are much more challenging to implement. On the other hand, Pauli measurements are experimental friendly, therefore, extremely important both theoretically and experimentally. 
If one is restricted to use Pauli measurements, Flammia, Gross, Liu and Eisert observed that the quantum state tomography can be accomplished using ${\mathcal{O}}(\frac{n\cdot 16^n}{\epsilon^2})$ copies in \cite{Flammia_2012}. This was improved the bound to ${\mathcal{O}}(\frac{n\cdot 12^n}{\epsilon^2})$ in  \cite{Gu__2020}.

In \cite{Cotler_2020}, Cotler and Wilczek introduced the problem called \textit{quantum overlapping tomography} problem. The goal is to output the classical description of all $k$-qubit reduced
density matrix of an $n$-qubit system. This problem is also of great importance in quantum information science, because many important physical quantities, for instance, energy and entropy, depend on very small parts of the whole system only, in other words, depends on the set of local reduced density matrices. By using the perfect hash families, they showed that one only needs to use $e^{\mathcal{O}(k)}\log^2 n$ rounds of parallel measurements to achieve this goal.
Ref.~\cite{Garc_a_P_rez_2020} proposed a measurement scheme to perform two-qubit tomography of all pairs. Later Ref.~\cite{evans2019scalable} provided an algorithm to estimate the expectation value of $m$ Pauli operators with weight $\leq k$ using $\mathcal{O}(3^k \log m)$ measurements for small $k$. All of these bounds can be achieved using 
Pauli measurements. 

\subsection{Our results}
In this paper, we study the power of Pauli measurements, the most important class of measurements, in two tomography related problems.

In the first part of this paper, we revisit the sample complexity of \textit{quantum state tomography} problem,
\begin{prob} \label{tomography}
Given an unknown $n$-qubit quantum mixed state $\rho$, 
the goal is to output density matrices $\sigma$ such that 
\begin{align}
||\rho-\sigma||_1<\epsilon
\end{align}
for a given $\epsilon>0$.
How many copies of $\rho$ are needed to achieve this goal, with high probability? 
\end{prob}
The sample complexity of this problem is
nearly resolved in the general joint measurement setting and independent measurement setting \cite{HHJ+16}. However, it is still unclear in the local measurement, in particular, the Pauli measurement setting. We show that
\begin{thm}
 The \textit{quantum state tomography} problem can be solved using $\mathcal{O}(\frac{10^n\log\frac{1}{\delta}}{\epsilon^2})$ copies of $\rho$ via Pauli measurement on each qubit to success with probability at least $1-\delta$.
 \end{thm}

Secondly, we study the sample complexity of the \textit{quantum overlapping tomography} problem:  ``How
  many copies of states are necessary and sufficient to reconstruct all the $k$-reduced density matrices of unknown $n$-qubit state $\rho$, to within additive error $\epsilon$, for constant $k$?'' More precisely, the formula of this problem is
\begin{prob} \label{localtomography}
Given an unknown $n$-qubit quantum mixed state $\rho_{1,2,\dots,n}$, and $1\leq k\leq n$,
the goal of quantum overlapping tomography is to output density matrices $\sigma_S$ for all $S\subseteq \{1,2,\dots,n\}$ with $|S|\leq k$ such that 
\begin{align}
||\rho_S-\sigma_S||_1<\epsilon
\end{align}
for a given $\epsilon>0$, where $\rho_S$ denotes the reduced density matrix of $\rho_{1,2,\dots,n}$ on $S$. 

If one only cares about $M$ density matrices $\sigma_{S_1},\sigma_{S_2},\cdots,\sigma_{S_M}$ with $|S_i|\leq k$, this problem is called partial quantum overlapping tomography.

How many copies of $\rho$ are needed to achieve this goal, with probability at least $1-\delta$ for $\delta>0$? 
\end{prob}

We show that,
\begin{thm} \label{localtomographyt}
The sample complexity of quantum overlapping tomography of $n$-qubit system is $\Theta(\epsilon^{-2}\cdot\log (n/\delta))$ for constant $k$ to succeed with probability at least $1-\delta$, even using general joint measurement schemes.
General quantum overlapping tomography problem for $1\leq k\leq n$ can be accomplished by performing Pauli measurements on $\mathcal{O}(\epsilon^{-2}\cdot 10^k\cdot \log({{n}\choose{k}}/\delta))$ copies. Moreover, for partial quantum overlapping tomography with $M$ outcomes, the sample complexity is $\mathcal{O}(\epsilon^{-2}\cdot 10^k\cdot \log(2M/\delta))$ using Pauli measurements.
\end{thm}

For the lower bound part, we show that to succeed with probability at least $1-\delta$, it
  is necessary to have $ C_\ell \frac{\log(n/\delta)}{\epsilon^{2}}$ copies even if joint measurement on many copies of $\rho$ are allowed, where $C_\ell>0$ is a constant.   
  For example, if $m$ copies of the $n$-qubit system are prepared and measured.
In the joint measure setting, the $mn$ qubits may be accessed collectively.

The upper bound 
 can be achieved by the following algorithm,

 \begin{algorithm}
Repeat the following measurement $\mathcal{O}(g(k)\cdot\epsilon^{-2}\cdot\log({{n}\choose{k}}/\delta))$ times\;
Measuring each qubit using some informationally complete measurement for any copy of $\rho$\;
\caption{Quantum Overlapping Tomography}
\end{algorithm}

Here $g(k)>0$ is a function depends on the informationally complete measurements and $k$ only.
Therefore, for constant $k$ and fixed informationally complete measurements, the used number of copies becomes $\mathcal{O}(\epsilon^{-2}\cdot\log (n/\delta))$.

As an example, we choose informationally complete measurement $\mathcal{M}=\frac{1}{6}\{\sigma_I+\sigma_X,\sigma_I-\sigma_X,\sigma_I+\sigma_Y,\sigma_I-\sigma_Y,\sigma_I+\sigma_Z,\sigma_I-\sigma_Z\}$, which can be regarded as random Pauli measurement, on each qubit. Then we can obtain the measurement scheme using $\mathcal{O}(\epsilon^{-2}\cdot 10^k\cdot \log(M/\delta))$ copies of $\rho$. 

This is the first nontrivial example that Pauli measurements are as powerful as general joint measurements asymptotically. For this example, the number of unknown parameters of the output is exponentially larger than the number of copies.

This implies that for all observables $O_S\otimes I_{\bar{S}}$ with $S$ being a set consisting of at most $k$-qubit, $\bar{S}$ being the complementary set of $S$ and $-I\leq O_S\leq I$, we can output estimation $o_S$ such that
\begin{align*}
|\tr[\rho(O_S\otimes I_{\bar{S}})]-o_S|\leq \eps
\end{align*}

\section{Preliminaries}

A positive-operator valued measure (POVM) on a finite dimensional Hilbert space $\mathcal{H}$ is a set of positive semi-definite matrices $\mathcal{M}=\{M_i\}$ such that
\begin{align*}
\sum M_i=I_{\mathcal{H}}.
\end{align*}

We need the concept of the informationally complete POVM as follows,
\begin{dfn}
An informationally complete POVM is a POVM whose outcome probabilities are sufficient to determine any state.
\end{dfn}

Equaivalently, a POVM $\mathcal{M}=\{M_i\}$ on $d$-dimensional $\mathcal{H}$ is informationally complete if the linear span of $\{M_i\}$ equals to the whole $d\times d$ matrix space. In qubit system, it means $\sigma_I,\sigma_X,\sigma_Y$ and $\sigma_Z$ all live in the linear span of $\{M_i\}$, where $\sigma_I$, $\sigma_X,\sigma_Y$ and $\sigma_Z$ are Pauli matrices,
\begin{align*}
\sigma_I=\begin{bmatrix}1 &0\\0&1\end{bmatrix}, \sigma_X=\begin{bmatrix}0 &1\\1&0\end{bmatrix},  \sigma_Z=\begin{bmatrix}1 &0\\0&-1\end{bmatrix}, \sigma_Y=\begin{bmatrix}0 &i\\-i&0\end{bmatrix}.
\end{align*}

It is direct to observe that $\mathcal{M}=\frac{1}{6}\{\sigma_I+\sigma_X,\sigma_I-\sigma_X,\sigma_I+\sigma_Y,\sigma_I-\sigma_Y,\sigma_I+\sigma_Z,\sigma_I-\sigma_Z\}$ is an informationally complete measurement.
\begin{observation}
Given sufficient measurement outcomes of an informationally complete POVM, one can determine the state with high accuracy and confidence.
\end{observation}

\begin{observation}
For information complete POVMs, $\mathcal{M}_1$ on $\mathcal{H}_1$ and $\mathcal{M}_2$ on $\mathcal{H}_2$, $\mathcal{M}_1\otimes\mathcal{M}_2$ is an informationally complete POVM on $\mathcal{H}_1\otimes\mathcal{H}_2$.
\end{observation}
Directly, $\mathcal{M}^{\otimes n}=\mathcal{M}=\frac{1}{6^n}\{\sigma_I+\sigma_X,\sigma_I-\sigma_X,\sigma_I+\sigma_Y,\sigma_I-\sigma_Y,\sigma_I+\sigma_Z,\sigma_I-\sigma_Z\}^{\otimes n}$ is an informationally complete POVM of $n$-qubit system.
\begin{dfn}
Let $X_1,X_2,\cdots,X_n$ be $n$ samples of a distribution on $\{1,2,\cdots,n\}$. Then the empirical distribution is defined as
\begin{align*}
\hat{p(i)}=\frac{\mathrm{number~of~occurrences~of}~i}{n}
\end{align*}

\end{dfn}

The following McDiarmid’s inequality \cite{mcdiarmid_1989} will be used in this paper.
\begin{lem}\label{mc}
Consider independent random variables ${\displaystyle X_{1},X_{2},\dots X_{n}}$ on probability space $ {\displaystyle (\Omega ,{\mathcal {F}},{\text{P}})}$ where ${\displaystyle X_{i}\in {\mathcal {X}}_{i}}$ for all ${\displaystyle i}$ and a mapping ${\displaystyle f:{\mathcal {X}}_{1}\times {\mathcal {X}}_{2}\times \cdots \times {\mathcal {X}}_{n}\rightarrow \mathbb {R} }$. Assume there exist constant $ {\displaystyle c_{1},c_{2},\dots ,c_{n}} $ such that for all $ {\displaystyle i}$,
\begin{widetext} 
\begin{align}{\displaystyle {\underset {x_{1},\cdots ,x_{i-1},x_{i},x_{i}',x_{i+1},\cdots ,x_{n}}{\sup }}|f(x_{1},\dots ,x_{i-1},x_{i},x_{i+1},\cdots ,x_{n})-f(x_{1},\dots ,x_{i-1},x_{i}',x_{i+1},\cdots ,x_{n})|\leq c_{i}.} 
\end{align}
\end{widetext}
(In other words, changing the value of the ${\displaystyle i}$-th coordinate ${\displaystyle x_{i}}$ changes the value of ${\displaystyle f}$ by at most ${\displaystyle c_{i}}$.) Then, for any ${\displaystyle \epsilon >0}$,
\begin{widetext} 
\begin{align} 
{\displaystyle {\mathrm{Pr}}(f(X_{1},X_{2},\cdots ,X_{n})-\mathbb {E} [f(X_{1},X_{2},\cdots ,X_{n})]\geq \epsilon )\leq \exp \left(-{\frac {2\epsilon ^{2}}{\sum _{i=1}^{n}c_{i}^{2}}}\right)} .
\end{align}
\end{widetext}
\end{lem}

\section{Quantum Tomography Using Pauli Measurement}

We start from the following observation:

When we measure an element of Pauli group, for instance $\sigma_X\sigma_Y$, on a two-qubit state $\rho$, the outcome is a sample from a $4$-dimensional probability distribution, says $(p_{00},p_{01},p_{10},p_{11})$, such that
\begin{align*}
\tr(\rho(\sigma_X\otimes \sigma_Y))=p_{00}-p_{01}-p_{10}+p_{11}.
\end{align*}

One can easily observe that
\begin{align*}
\tr(\rho(\sigma_X\otimes \sigma_I))=p_{00}+p_{01}-p_{10}-p_{11},\\
\tr(\rho(\sigma_I\otimes \sigma_Y))=p_{00}-p_{01}+p_{10}-p_{11},\\
\tr(\rho(\sigma_I\otimes \sigma_I))=p_{00}+p_{01}+p_{10}+p_{11}.
\end{align*}

In other words, by measuring $XY$, we actually obtained a sample of $\sigma_X\sigma_I$, a sample of $\sigma_Y\sigma_I$, and a sample of $\sigma_I\sigma_I$. 

This is also true for general $n$-qubit system as the following observation shows,
\begin{observation}
For any $P=P_1\otimes P_2\otimes\cdots\otimes P_n\in\{\sigma_X,\sigma_Y,\sigma_Z\}^{\otimes n}$, the measurement result of performing measurement $P_i$ on the $i$-th qubit is an $n$-bit string $s$. One can interpretes the measurement outcome of performing $Q_i\in\{\sigma_I,\sigma_X,\sigma_Y,\sigma_Z\}$ on the $i$-th qubit if $Q_i=P_i$ or $Q_i=\sigma_I$. We call these $Q=Q_1\otimes Q_2\otimes\cdots\otimes Q_n$ corresponds to $P$.
\end{observation}
Our measurement scheme is as follows: For any $\epsilon>0$, fix an integer $m=16\cdot\frac{10^n\log\frac{1}{\delta}}{3^n\cdot \epsilon^2}$.

For any $P\in\{\sigma_X,\sigma_Y,\sigma_Z\}^{\otimes n}$, one performs $m$ times $P$ on $\rho$, and records the $m$ samples of the $2^n$ dimensional outcome distribution.

According to the key observation, this measurement scheme provides $m\cdot 3^{n-w}$ samples of the expectation $\tr(\rho P)$, says $\frac{\mu_P}{m\cdot 3^{n-w}}$ for each Pauli operator $P\in \{\sigma_I,\sigma_X,\sigma_Y,\sigma_Z\}^{\otimes n}$ with weight $w$, where $-m\cdot 3^{n-w}\leq\mu_P\leq m\cdot 3^{n-w}$.

Output 
\begin{align*}
\sigma=\sum_P \frac{\mu_P}{m\cdot 3^{n-w}\cdot 2^n} P.
\end{align*}

Using this scheme, we obtained $m\cdot 3^n$ independent samples,
\begin{align*}
X_1,X_2,\cdots, X_{m\cdot 3^n}
\end{align*}
where each $0\leq X_i\leq 2^n-1$.

Fruthermore, $X_1,X_2,\cdots,X_{m}$ corresponds to the measurement $\sigma_X^{\otimes n}$; $X_{m+1},X_{m+2},\cdots,X_{2m}$ corresponds to the measurement $\sigma_X^{\otimes n-1}\otimes \sigma_Y$; $\cdots$

It is direct to observe that, for any $P$ of weight $w$,
$\mu_P=\sum_{j=0}^{m\cdot 3^{n-w}-1} Z_j$,
where $Z_j$ are independent samples from distribution $Z$
\begin{align*}
\mathrm{Pr}(Z=1)=\frac{1+\tr (\rho P)}{2}\\
\mathrm{Pr}(Z=-1)=\frac{1-\tr (\rho P)}{2}.
\end{align*}
Furthermore, $Z_j$ can be obtained from samples
\begin{align*}
X_1,X_2,\cdots, X_{m\cdot 3^n}.
\end{align*}

Therefore,
\begin{align*}
\sigma=\sum_P \frac{\mu_P}{m\cdot 3^{n-w_P}\cdot 2^n} P.
\end{align*}
is defined according to samples
\begin{align*}
X_1,X_2,\cdots, X_{m\cdot 3^n}.
\end{align*}
It is direct to verify that
\begin{align*}
\mathbb{E}\sigma=\rho,
\end{align*}
where the expectation is taken over the probabilistic distribution according to the measurements.

For any 
\begin{align*}
\rho=\sum_P \frac{\alpha_P}{2^n} P
\end{align*}
we can define the function $f:X_1\times X_2\times \cdots\times X_{m\cdot 3^n}\mapsto \mathbb{R}$
\begin{align*}
f=||\rho-\sigma||_2
\end{align*}
According to Cauchy inequality, we have
\begin{align*}
&\mathbb{E}f \\
\leq& \sqrt{\mathbb{E}f^2}\\
=& \sqrt{\mathbb{E} (\tr\rho^2-2\tr \rho\sigma+\tr\sigma^2)}\\
=& \sqrt{\mathbb{E}\tr\sigma^2-\tr\rho^2}\\
=& \sqrt{\frac{1}{2^n}\sum_P (\mathbb{E} \frac{\mu_P^2}{m^2\cdot 9^{n-w_P}}-\alpha_P^2)}\\
=&\sqrt{\frac{1}{m\cdot 2^n}\cdot{\sum_P \frac{1-\alpha_P^2}{3^{n-w_P}}}}\\
<&\sqrt{\frac{1}{m\cdot 2^n}\cdot{\sum_P \frac{1}{3^{n-w_P}}}}\\
=& \sqrt{\frac{1}{m\cdot 2^n}\cdot{\sum_{w_P=0}^n \frac{1}{3^{n-w_P}}{{n}\choose{w_P}}3^{w_P}}}\\
=&\sqrt{\frac{1}{m\cdot 6^n}\cdot\sum_{w_P=0}^n (1+9)^n}\\
=& \sqrt{\frac{5^n}{m\cdot 3^n}}\\
<& \frac{\epsilon}{2\cdot \sqrt{2^n}}.
\end{align*}

For any sample $X_i$ corresponding to $P\in\{\sigma_X,\sigma_Y,\sigma_Z\}^{\otimes n}$. If only $X_i$ is changed, it would only change $\mu_Q$ for those $Q\in\{\sigma_I,\sigma_X,\sigma_Y,\sigma_Z\}^{\otimes n}$ where
$Q$ is obtained by replacing some $\{\sigma_X,\sigma_Y,\sigma_Z\}$s of $P$ by $\sigma_I$.
Moreover, those $\mu_Q$ would change at most $2$.
According to triangle inequality, $f$ would change at most 
\begin{align*}
&||\sum_Q  \frac{\Delta\mu_Q}{m\cdot 3^{n-w_Q}\cdot 2^n} Q||_2\\
=&\sqrt{\sum_Q \frac{\Delta\mu_Q^2}{m^2\cdot 9^{n-w_Q}\cdot 2^n}}\\
\leq& \sqrt{\sum_Q \frac{2^2}{m^2\cdot 9^{n-w_Q}\cdot 2^n}}\\
=& \sqrt{\sum_{w_Q=0}^n \frac{2^2}{m^2\cdot 9^{n-w_Q}\cdot 2^n}{{n}\choose{w_Q}}}\\
=&\frac{2\cdot \sqrt{5}^n}{m\cdot 3^n},
\end{align*}
where $Q$ ranges over all Paulis which corresponds to $P$, and $\delta\mu_Q$ denotes the difference of $\mu_Q$ when $X_i$ is changed.

For $\epsilon>0$, let $\epsilon'=\frac{\epsilon}{\sqrt{2^n}}$.
We have 
\begin{align*}
&\mathrm{Pr}(||\rho-\sigma||_1>\epsilon)\\
\leq& \mathrm{Pr}(||\rho-\sigma||_2>\epsilon')\\
=& \mathrm{Pr}(f>\epsilon')\\
\leq & \mathrm{Pr}(f-\mathbb{E}f>\frac{\epsilon'}{2})\\
<&\exp(-\frac{\epsilon'^2}{4 \cdot\frac{4\cdot 5^n}{m^2\cdot 9^n}\cdot m\cdot 3^n})\\
=&\exp(-\frac{m\cdot 3^n\cdot\epsilon'^2}{16\cdot 5^n})\\
<&\delta, 
\end{align*}
where the first step is by the relation between the trace norm and 2 norm; the third step is by $\mathbb{E}f\leq \frac{\epsilon'}{2}$; the fourth step is by Lemma \ref{mc} (McDiarmid’s inequality).

The total number of used copies is
\begin{align*}
m\cdot 3^n=16\cdot\frac{10^n\log\frac{1}{\delta}}{\epsilon^2}.
\end{align*}

\section{Quantum Overlapping Tomography}
In this section, we study the quantum overlapping tomography. 

We first analyze Algorithm 1.

Fixing $n$ informationally complete POVM of one-qubit system, says $\mathcal{M}_1,\mathcal{M}_2,\cdots,\mathcal{M}_n$. For any given unknown $n$-qubit state $\rho$, we perform $\mathcal{M}_1\otimes\mathcal{M}_2\otimes\cdots\otimes\mathcal{M}_n$ on $\rho$ and obtain a string $s_1s_2\cdots s_n$ with $s_i$ denoting the measurement outcome of $\mathcal{M}_i$. Assume the output probability distribution is $p_{1,2,\cdots,n}$. we know that $s_1s_2\cdots s_n$ obeys the distribution of $p_{1,2,\cdots,n}$.

The key observation is that measuring each qubit independently preserves the local structure. More precisely,
\begin{observation}
For any $S=\{i_1,i_2,\cdots,i_r\}\subseteq\{1,2,\cdots,n\}$, the distribution of $s_{i_1}s_{i_2}\cdots s_{i_r}$ obeys the outcome distribution of performing $\mathcal{M}_{i_1}\otimes\mathcal{M}_{i_2}\otimes\cdots\otimes\mathcal{M}_{i_r}$ on $\rho_S$. Moreover,  $s_{i_1}s_{i_2}\cdots s_{i_r}$ obeys the marginal distribution $p_{S}$.
\end{observation}

By Algorithm 1, we repeat the measurement $\mathcal{M}_1\otimes\mathcal{M}_2\otimes\cdots\otimes\mathcal{M}_n$ $m$ times to obtain $m$ samples of $p_{1,2,\cdots,n}$. For any sets $S_1,S_2,\cdots,S_M\subseteq\{1,2,\cdots,n\}$, Obervation 3 implies that we have $m$ samples of the marginal distributions $p_{S_1}$, $p_{S_2}\cdots$, and $p_{S_M}$, not independent. Although, $m$ samples maybe not enough to recover $p_{1,2,\cdots,n}$, but enough to recover $p_{S_i}$ for each $S_i$ with high accuracy and high confidence using Chernoff bound. According to Observation 1 and Observation 2, this implies one can recover $\rho_{S_i}$ with high accuracy, trace distance error $\epsilon$, and sucessful probability at least $1-\delta'$ for each $S_i$ using $m=\mathcal{O}(g(k)\cdot\frac{\log \delta'}{\epsilon^2})$.

To successfully recover $\rho_{S_i}$ simultaneously, we only need the recovery of each $\rho_{S_i}$ with high accuracy and probability at least $1-\frac{\delta}{{{n}\choose{k}}}$ according to union bound.

Therefore, we only need $m=\mathcal{O}(g(k)\cdot\frac{\log {{{n}\choose{k}}}/{\delta}}{\epsilon^2})$ copies.
\subsection{Joint Measurement Lower bound}

In this subsection, we show that $\Omega(\frac{\log(n/\delta)}{\epsilon^2})$ copies are necessary to quantum overlapping tomography by proving that: $\Omega(\frac{\log(n/\delta)}{\epsilon^2})$ copies are necessary for quantum overlapping tomography with $k=1$, even if general \textit{joint measurements} are used.

Because trace distance is non-increasing according to partial trace, we can conclude that any measurement scheme can solve the quantum overlapping tomography problem for $k\geq 1$, automatically, it solves the case that $k=1$. Moreover, to deal with general joint measurement schemes, we focus on classical distributions. 

We first consider a simple question: Given a binary random variable $X$ obeys distribution $q_0=(1/2-\epsilon,1/2+\epsilon)$ or $q_1=(1/2+\epsilon,1/2-\epsilon)$. The goal is to find out which distribution is the true distribution. It is widely known that: For any fixed $m$ as the number of tossing this coin, the best strategy is to toss the coin $m$ times and declare the index ($0$ or $1$) that appears less.

Let the $X_1,X_2,\dots,X_m$ be the $m$ samples of $q_1$ and any $0\leq t\leq 2m\epsilon$, Ref.~\cite{Mousavi2016} proved the following:
\begin{align}
\mathrm{Pr}(\sum_{i=1}^m X_i>t+m(1/2-\epsilon))\geq \frac{1}{4}\cdot \exp(-\frac{2t^2}{m(1/2-\epsilon)})
\end{align}
By choosing $t=m\epsilon$, $\mathrm{Pr}(\sum_{i=1}^m X_i>t+m(1/2-\epsilon))=\mathrm{Pr}(\sum_{i=1}^m X_i>m/2)$ is the probability of answering $q_0$, a lower bound of the failure probability.

To success with probability at least $1-\delta'$, we must have
\begin{align}
\delta'\geq\frac{1}{4}\cdot \exp(-\frac{2m\epsilon^2}{1/2-\epsilon})
\end{align}
That is, $\frac{1-2\epsilon}{4\epsilon^2}\log(\frac{1}{4\delta'})$ samples are needed to distinguish $q_0$ and $q_1$ with confidence at least $1-\delta'$. 

Back to our problem of showing $\Omega(\frac{\log(n/\delta)}{\epsilon^2})$ samples of $n$-qubit state $\rho$ are necessary to solve the quantum overlapping tomography for $k\geq 1$, to within additive error $\epsilon$ and confidence at least $1-\delta$, we consider the classical distributions $p_{i_1,i_2,\cdots,i_n}=q_{i_1}\otimes q_{i_2}\otimes\cdots \otimes q_{i_n}$ where $q_0=(1/2-\epsilon,1/2+\epsilon)$ and $q_1=(1/2+\epsilon,1/2-\epsilon)$. In total, there are $2^n$ different distributions.

Suppose there is a quantum procedure $\mathcal{A}$ which uses $m$ copies of $\rho$ to accomplish the quantum overlapping tomography with probability at least $1-\delta$. 

Let $Z_1,Z_2,\cdots,Z_n$ be random variables which obey uniform binary distribution.
Choose each $p_{Z_1,Z_2,\cdots,Z_n}$ with probability $1/2^n$, and apply $\mathcal{A}$ on $m$ copies of $p_{Z_1,Z_2,\cdots,Z_n}$. Because the $\ell_1$ norm is non-increasing under partial trace, we know that according to the output of $\mathcal{A}$, we can sucessfully recover the indices $Z_1,Z_2,\cdots,Z_n$ with probability at least $1-\delta$.

In the following, we first observe that quantum procedure does not help in recovering $Z_1,Z_2,\cdots,Z_n$ from samples of $p_{Z_1,Z_2,\cdots,Z_n}$.
We assume the joint measurement $(M_{0,0,\cdots,0},\cdots,M_{1,1,\cdots,1})$ applied on $m$ copies (samples) of $p$ such that the measurement outcome $M_{i_1,i_2,\cdots,i_n}$ allows us to answer $Z_1=i_1,Z_2=i_2,\cdots,Z_n=i_n$. Here $M_{i_1,i_2,\cdots,i_n}$s are $2^{mn}\times 2^{mn}$ matrices.

We first observe that $p_{Z_1,Z_2,\cdots,Z_n}$s are all diagonal, so are  $p_{Z_1,Z_2,\cdots,Z_n}^{\otimes m}$.
Then, the off diagonal elements of $M_{i_1,i_2,\cdots,i_n}$ has no effect for this task.
Therefore, we only need to consider the procedure in the following two steps: At the first step, measure $m$ copies of $p_{Z_1,Z_2,\cdots,Z_n}$s in the diagonal basis; at the second step, output according to some probability distributions.

The first step ensures that we only need to measure each copy of $p_{Z_1,Z_2,\cdots,Z_n}$s in the diagonal basis, as there is no difference. By the convexity of the successful probability, we know that deterministic function works best in the second step, that is declare the index ($0$ or $1$) that appears less for each $1\leq j\leq n$.

Now, we assume the output random virable is $Y_1,Y_2,\cdots,Y_n$, our goal is
\begin{align}
\mathrm{Pr}(Y_1=Z_1,Y_2=Z_2,\cdots,Y_n=Z_n)\geq 1-\delta.
\end{align}
By Bayes' theorem, we know that
\begin{align*}
&\mathrm{Pr}(Y_1=Z_1,Y_2=Z_2,\cdots,Y_n=Z_n)\\
=&\mathrm{Pr}(Y_1=Z_1)\times \mathrm{Pr}(Y_2=Z_2|Y_1=Z_1)\times\cdots\times \mathrm{Pr}(Y_n=Z_n|Y_1=Z_1,\cdots,Y_{n-1}=Z_{n-1})\\
\leq& (1-\delta')\cdot (1-\delta')\cdot\cdots\cdot(1-\delta')\\
=& (1-\delta')^m,
\end{align*}
where we use the fact that $p_{Z_1,Z_2,\cdots,Z_n}$s are all in tensor product form, therefore, $\mathrm{Pr}(Y_2=Z_2|Y_1=Z_1)=\mathrm{Pr}(Y_2=Z_2)$, and so on.
$(1-\delta')$ denotes the sucessful probability of discriminate $q_0$ and $q_1$ with $m$ copies.

Therefore, we require that 
\begin{align*}
(1-\delta')^n>1-\delta.
\end{align*}
That is $\delta'=\Theta(\frac{\delta}{n})$, this implies the bound $m=\Omega(\frac{\log(n/\delta)}{\epsilon^2})$.

\subsection{Pauli Measurement Upper bound}
In this subsection, we analyze Algorithm 1 using the measurement scheme $\{M_0,M_1,M_2,M_3,M_4,M_5\}^{\otimes n}$ as a special case, where
\begin{align*}
M_0=\frac{\sigma_I+\sigma_X}{6}, \ \ \ \ \ M_1=\frac{\sigma_I-\sigma_X}{6},\\
M_2=\frac{\sigma_I+\sigma_Y}{6},\ \ \ \ \ M_3=\frac{\sigma_I-\sigma_Y}{6},\\
M_4=\frac{\sigma_I+\sigma_Z}{6},\ \ \ \ \ M_5=\frac{\sigma_I-\sigma_Z}{6}.
\end{align*}
This can be regarded as a random chosen of $P\in\{\sigma_X,\sigma_Y,\sigma_Z\}^{\otimes}$ and measure $\rho$ in the basis of $P$, and repeat it $m$ times. That is, 

 \begin{algorithm}
Repeat the following measurement $32\cdot 10^k\cdot\epsilon^{-2}\cdot\log(2{{n}\choose{k}}/\delta)$ times\;
For $i=1$ to $n$:
measure the $i$-th qubit in a random chosen basis from$\{\sigma_X,\sigma_Y,\sigma_Z\}$\;
\caption{Quantum Overlapping Tomography by Pauli Measurements}
\end{algorithm}

We observe the following, 
\begin{observation}
For any $S=\{i_1,i_2,\cdots,i_k\}\subseteq\{1,2,\cdots,n\}$, 
with probability at least $1-\frac{2^m}{e^m}\cdot 3^k$, each $P\in \{\sigma_X,\sigma_Y,\sigma_Z\}^{\otimes k}$ was measured for $\rho_S$ at least $m=16\cdot \frac{10^k\cdot\log(2{{n}\choose{k}}/\delta)}{3^k\cdot\epsilon^{2}}$ times.
\end{observation}
\begin{proof}
It is direct to observe that for any $P\in \{\sigma_X,\sigma_Y,\sigma_Z\}^{\otimes k}$, it was chosen as a basis to measure is no more than $m$ times with probability at most
\begin{align*}
&\sum_{i=0}^m {{2m\cdot 3^k}\choose{i}}(1-\frac{1}{3^k})^{2m\cdot 3^k-i} \frac{1}{3^{ki}}\\
=&(1-\frac{1}{3^k})^{2m\cdot 3^k}\cdot\sum_{i=0}^m {{2m\cdot 3^k}\choose{i}}(1-\frac{1}{3^k})^{-i} \frac{1}{3^{ki}}\\
= & (1-\frac{1}{3^k})^{2m\cdot 3^k}\cdot\sum_{i=0}^m {{2m\cdot 3^k}\choose{i}} \frac{1}{(3^{k}-1)^i}
\end{align*}
The ratio of the $i+1$-th term and the $i$-th term is 
\begin{align*}
\frac{1}{3^k}\cdot\frac{2m\cdot 3^k-i}{i+1}\geq 2
\end{align*}
Therefore, the serie grows faster than a geometric progression with common ratio $2$. Then,
\begin{align*}
&(1-\frac{1}{3^k})^{2m\cdot 3^k}\cdot\sum_{i=0}^m {{2m\cdot 3^k}\choose{i}} \frac{1}{(3^{k}-1)^i}\\
\leq &2\cdot(1-\frac{1}{3^k})^{2m\cdot 3^k}{{2m\cdot 3^k}\choose{m}} \frac{1}{(3^{k}-1)^m}\\
\leq & 3\sqrt{\frac{2\cdot 3^k}{2\pi\cdot m\cdot (2\cdot 3^k-1)}}[2(1-\frac{1}{2\cdot 3^k-1})^{2\cdot 3^k-1}]^m\\
\leq & 3\sqrt{\frac{2\cdot 3^k}{2\pi\cdot m\cdot (2\cdot 3^k-1)}}\frac{2^m}{e^m}\\
\leq &\frac{2^m}{e^m}.
\end{align*}
where we use the bound that
\begin{align*}
\sqrt{2\pi n} \cdot n^n\cdot e^{-n}\cdot e^{\frac{1}{12n}}< n!<\sqrt{2\pi n} \cdot n^n\cdot e^{-n}\cdot e^{\frac{1}{12n+1}}
\end{align*}
and the number series
\begin{align*}
(1-\frac{1}{d})^d
\end{align*}
is increasing and goes to $\frac{1}{e}$.

Then, by union bound, we know that for any $S$ of size $k$,with probability at least $1-\frac{2^m}{e^m}\cdot 3^k$, each $P\in \{\sigma_X,\sigma_Y,\sigma_Z\}^{\otimes k}$ was measured for $\rho_S$ at least $m$ times.
\end{proof}

According to Theorem 1, the tomography of $\rho_S$ with trace distance error $\epsilon$ was successful with probability at least 
\begin{align*}
1- [3^k\cdot \frac{2^m}{e^m}+\frac{\delta}{2{{n}\choose{k}}}]> 1-\frac{\delta}{{{n}\choose{k}}},
\end{align*}
the last inequality follows from
\begin{align*}
&3^k\cdot \frac{2^m}{e^m}\\
< & 3^k (\frac{2}{e})^{12\cdot 3^k}\cdot (\frac{2}{e})^{4\cdot 3^k\cdot \log(2{{n}\choose{k}}/\delta)}\\
< & 3^k (\frac{2}{e})^{12\cdot k}\cdot (\frac{1}{e})^{ \log(2{{n}\choose{k}}/\delta)}\\
<& \frac{\delta}{2{{n}\choose{k}}}.
\end{align*}
By union bound, the quantum overlapping tomography was with trace distance error $\epsilon$ was successful with probability at least 
\begin{align*}
1- {{n}\choose{k}}\cdot \frac{\delta}{{{n}\choose{k}}}=1-\delta
\end{align*}

\section{Discussion}

For the general quantum state tomography problem, our measurement scheme is much more efficient than previous schemes in the Pauli measurements setting. Our result raises an important question on the
performance of {\em local} measurements, in particular Pauli measurements, in which each qubit
is measured at a time.

The sample complexity of the quantum overlapping tomography problem is
nearly resolved here up to a constant factor. 

An independent work~\cite{Huang_2020} 
analyzes random Clifford measurement and random Pauli measurement
and proves that it only requires $\mathcal{O}(\frac{\log(M)\cdot s}{\epsilon^2})$ copies
to achieve $\epsilon$ accuracy estimation for $\tr(O_1\rho),\cdots,\tr(O_M\rho)$, a more general question, where $s$ depends on the structure of $O_i$s. They also show that this is tight if only single-copy measurements are allowed. If $O_i$s are $k$-qubit Paulis, $\mathcal{O}(\frac{\log(M/\delta)\cdot 3^k}{\epsilon^2})$ which coinsides with the bound provided in Ref.~\cite{evans2019scalable}. 
These upper bound does not cover ours, because a $4^k$ coefficient would be added when we move from Pauli estimation to state tomography, thus, these results would lead to a $\mathcal{O}(\frac{\log(M/\delta)\cdot 12^k}{\epsilon^2})$ measurement scheme. 
Also, this lower bound does not imply our lower bound, for the specific quantum overlapping tomography problem, because this bound is not for joint measurement and does not consider the successful probability parameter.

\section{Acknowledgments}
We thank Steve Flammia for the discussion on quantum overlapping tomography. We thank Richard Keung
for clasifying their bound in \cite{Gu__2020}.
This work was supported by DE180100156.

\bibliography{opt-tomo}

\end{document}